\documentclass[11pt]{article}
\usepackage{amsmath}
\usepackage{amsthm}
\usepackage{amssymb}
\usepackage{dsfont}
\usepackage{hyperref}
\usepackage{cleveref}
\usepackage{mathrsfs}
\usepackage{mathtools}
\usepackage{url}
\usepackage{xcolor}

\newtheorem{theorem}{Theorem}[section]
\newtheorem{proposition}{Proposition}[section]
\newtheorem{lemma}{Lemma}[section]
\newtheorem{corollary}{Corollary}[section]

\theoremstyle{definition}
\newtheorem{definition}{Definition}[section]
\newtheorem{example}{Example}[section]
\newtheorem{remark}{Remark}[section]

\newtheorem*{assumption*}{Assumption}
\newtheorem{notation}{Notation}[section]

\newcommand{\bigO}{\mathit{O}}
\newcommand{\smallo}{\mathit{o}}

\newcommand{\Prob}{\mathds{P}}

\newcommand{\R}{\mathbb{R}}

\newcommand{\Z}{\mathbb{Z}}
\newcommand{\N}{\mathbb{N}}
\newcommand{\F}{\mathbb{F}}

\newcommand{\Id}{\textnormal{I}}

\newcommand{\G}{\mathcal{G}}    
\newcommand{\M}{\mathfrak{M}}    
\newcommand{\Num}{\mathcal{N}}    
\newcommand{\Fix}{\textnormal{Fix}}    
\newcommand{\code}{\mathcal{C}}

\newcommand{\dist}{\textnormal{d}}
\newcommand{\weight}{\textnormal{w}}

\newcommand{\underm}{\underline{m}}
\newcommand{\overm}{\overline{m}}

\title{\textbf{The asymptotic number of equivalence classes \\ of linear codes with given dimension}}

\usepackage{authblk}

\author{Andrea Di Giusto\thanks{A.D.G. is supported by the European Commission through MSCA-DN project ENCODE.}, and Alberto Ravagnani\thanks{A.R. is supported by the Dutch Research Council NWO through grants OCENW.KLEIN.539 and VI.Vidi.203.045, by the European Commission via the ENCODE MSCA-DN project, by the EuroTech Alliance, and by the Royal Academy of Arts and Sciences of the Netherlands.}}
\affil{Eindhoven University of Technology, the Netherlands}
\date{}

\usepackage[margin=3cm]{geometry}

\begin{document}

\setlength{\parindent}{20pt}

\begin{sloppypar}

\maketitle

\begin{abstract}
   We investigate the asymptotic number of equivalence classes of linear codes with prescribed length and dimension. While the total number of inequivalent codes of a given length has been studied previously, the case where the dimension varies as a function of the length has not yet been considered. We derive explicit asymptotic formulas for the number of equivalence classes under three standard notions of equivalence, for a fixed alphabet size and increasing length. Our approach also yields an exact asymptotic expression for the sum of all $q$-binomial coefficients, which is of independent interest and answers an open question in this context. Finally, we establish a natural connection between these asymptotic quantities and certain discrete Gaussian distributions arising from Brownian motion, providing a probabilistic interpretation of our results.
\end{abstract}

\medskip

\section{Introduction}

Studying error-correcting codes up to equivalence is a well-established practice in coding theory. Equivalent codes share all the properties that are relevant for a particular application, and can \textit{de facto} be used interchangeably. When focusing on linear block codes endowed with the Hamming metric, there exist three main notions of equivalence, each corresponding to a group action: \textit{permutation} equivalence, \textit{monomial} equivalence (probably the most popular), and \textit{semilinear} equivalence.
For binary block codes, these three notions coincide.

A natural question in coding theory is to count the number 
of $q$-ary codes that satisfy a particular property. When such property is having given \textit{length} $n$, \textit{dimension} $k$, and \textit{minimum distance} at least $d$,
the problem is equivalent to computing the parameters of certain lattices that are notoriously difficult to analyse~\cite{dowling1971codes,bonin1993automorphism,bonin1993modular,ravagnani2022whitney}. When working modulo code equivalence, one natural question is to compute the number of equivalence classes of $q$-ary codes with given length $n$ and dimension $k$, which is yet another impossible task.
This paper addresses the asymptotic version of this problem, solving it entirely for some parameters ranges. 

The problem of counting codes up to equivalence 
has been considered before in the coding theory literature. However, all references we are aware of focus on estimating the number of equivalence classes of codes over a given alphabet and with given length, \textit{without} fixing the code dimension. This paper fills in this gap by estimating the number of equivalence classes of codes whose dimension is $k$, and where $k=k(n)$ is a function of the length $n$. One of our main results determines the exact asymptotic behaviour of the number of equivalence classes for $q$ fixed and $n$ growing. More precisely, in~\Cref{thm:asymptotic_Nk} we show that, for sufficiently well-behaved dimension functions $k(n)$, the asymptotic numbers of 
permutation, monomial, and semilinear equivalence classes are
 $$\cfrac{q^{k(n)(n-k(n))}}{K_qn!},\qquad \cfrac{q^{k(n)(n-k(n))}}{K_qn!(q-1)^{n-1}},\qquad \cfrac{q^{k(n)(n-k(n))}}{K_qhn!(q-1)^{n-1}},$$
 respectively. We then relate these quantities to
the (better studied) number of equivalence classes of codes with unrestricted dimension. Interestingly, for a suitable choice of $k(n)$ the relation can be expressed very naturally using the Gaussian $\theta_2$ and $\theta_3$ distributions, which  control the dynamics of the Brownian motion. 

Our results are not necessarily related to the theory of error-correcting codes, and include asymptotic estimates of quantities that are of interest also in other fields, such as the ${q\textnormal{-binomial}}$ coefficients.
As a byproduct of our analysis,
we also determine the exact asymptotic behaviour of the sum of all $q$-binomials for fixed $q$ and $n\rightarrow\infty$,
answering an open question raised by Wild in~\cite{wild2000asymptotic}.

Before presenting the structure of the paper, we briefly survey the contributions to the problem made so far and introduce the various players.
The study of 
the asymptotic number of monomially inequivalent binary codes was initiated
by Wild in \cite{wild2000asymptotic}.
His main statement was correct, but 
a gap in the proof was found by Lax in \cite{lax2004character}.
A correct proof was later published by Hou in \cite{hou2007asymptotic}.
The state-of-the-art reference on monomial and permutation equivalence classes of $q$-ary linear codes is~\cite{hou2005asymptotic}, where Hou shows that the number $\Num_{n}$ of monomially inequivalent $q$-ary linear codes of length $n$ satisfies
\begin{equation}\label{eq:N_nq}
    \Num_{n}\sim\frac{\sum_{j=0}^n\binom{n}{j}_q}{n!(q-1)^{n-1}} \quad \mbox{for $n$ large.}
\end{equation}
However, the paper does not address the case where the dimension is restricted to a specific function $k(n)$ of the length.
For semilinear classes and a sum-up of the three notions of equivalence, see \cite{hou2009asymptotic}, also by Hou.
The focus of this paper is the quantity $\Num_{k,n}$, counting the number of inequivalent codes of a given dimension $k=k(n)$.

\paragraph*{Outline.}
The remainder of this paper is organized as follows. In~\Cref{sec:preliminaries} we establish the necessary background and formally state the problem studied in this paper.
In~\Cref{sec:asymptotic_number_eq_classes} we show the asymptotic relation between the number of inequivalent codes and $q$-binomial coefficients; in particular, we outline the limitations needed on the function $k(n)$ describing the dimension of the equivalence classes.
\Cref{sec:asymptotics_1} is devoted to the study of the $q$-binomial coefficient \smash{$\binom{n}{k(n)}_q$} as $n$ grows, and consequently to the description of the asymptotic number of inequivalent codes of dimension $k(n)$.
We then turn to the study of the proportion between this number and the number of all equivalence classes (without restriction on the dimension) in~\Cref{sec:asymptotics_2}.
We explain the link between these numbers and the Gaussian $\theta_2$ and~$\theta_3$ distributions,
and offer an asymptotic description of the sum of all $q$-binomials.

\section{Preliminaries and problem statement}\label{sec:preliminaries}
We start by establishing the notation for this paper and by stating the problem we are interested in. In the sequel, the closed interval with extrema $a,b\in\R$ is denoted by $[a,b]$; for $b\geq1$, we let~$[b]=[1,b]$.
For $x\in\R$, $\lfloor x\rfloor$ (resp. $\lceil x\rceil$) is the greatest (resp. smallest) $z\in\Z$ such that~$z\leq x$~($z\geq x$).
Throughout the paper we follow Bachmann-Landau notation notation ($\sim$, $\bigO$, $\smallo$) for asymptotic estimates; see~\cite{de2014asymptotic}.
All asymptotics are for $n\rightarrow\infty$, unless otherwise specified.

We include the coding theory background needed to read this paper; we refer to~\cite{huffman2010fundamentals} for further details.
With $q$ we always denote a prime power, $\F_q$ is the field with $q$ elements, and we let  $\F_q^*=\F_q\setminus\{0\}$.
We often omit~$q$ from the notations: the reader can 
assume it is fixed, unless otherwise specified.

\begin{definition}
    A (\textbf{linear}) \textbf{code} is a vector subspace $\code$ of $\F_q^n$. The \textbf{dimension} of $\code$ is its dimension over $\F_q$ as a linear space.
The \textbf{Hamming weight} of a vector $x=(x_1,\ldots,x_n)\in\F_q^n$ is the number of its nonzero coordinates: $\weight(x)=|\{i\in[n]\mid x_i\neq0\}|$.
The (\textbf{Hamming}) \textbf{distance} between $x,y\in\F_q^n$ is the Hamming weight of their difference:~$\dist(x,y)=\weight(x-y)$.
\end{definition}

The above notions are central in characterizing the error correcting capabilities of a code: the \textbf{minimum distance} of a nonzero code $\code$ is
\begin{equation*}
    \dist(\code)=\min\{\dist(c_1,c_2)\mid c_1,c_2\in\code, \, c_1\neq c_2\} = \min\{\weight(c)\mid c\in\code,\, c\neq0\},
\end{equation*}
and the maximum amount of errors that a code $\code$ can correct is $\lfloor(\dist(\code)-1)/2\rfloor$.

For $0\leq k\leq n$, the \textbf{Grassmannian} $\G(k,n)$ is the set of all codes of dimension $k$ in $\F_q^n$; its cardinality is the $q$-binomial coefficient $n$-choose-$k$.
The \textbf{projective space} is the union of all Grassmannians $\G(n)=\cup_{k=0}^n\G(k,n)$; its cardinality, which is the sum of $|\G(k,n)|$ for $k$ from 0 to $n$, is denoted by $S(n)$.
In formul\ae, we have
\begin{equation}\label{eq:q_binomial_def}
    |\G(k,n)|=\binom{n}{k}_q=\prod_{j=0}^{k-1}\frac{q^{n-j}-1}{q^{k-j}-1}\quad\textnormal{and}\quad S(n)=\sum_{k=0}^n\binom{n}{k}_q.
\end{equation}

It is natural to group codes in equivalence classes and thus ask the following question: when are two codes \textit{essentially the same}?
Informally speaking, we want two equivalent codes to be able to carry the same amount of information and to have the same error-correcting capabilities.
There are three types of equivalence in coding theory: \textit{permutation}, \textit{monomial} and \textit{semilinear} equivalence.
The equivalence classes are the orbits of the action of three groups of transformations of $\F_q^n$, respectively, whose names correspond to the equivalence type they describe.
We abuse terminology and call the equivalence classes also \textit{inequivalent codes}.

The \textbf{permutation} group $\mathfrak{S}_n$ is formed by all the $n\times n$ permutation matrices, and the \textbf{monomial} group $\M_n$ is the subgroup of $\textnormal{GL}(\F_q^n)$ generated by $\mathfrak{S}_n$ and all diagonal matrices.
These two groups inherit their action on $\F_q^n$ from $\textnormal{GL}(\F_q^n)$, and since the image of a subspace via a linear transformation is a subspace of the same dimension, the action extends to $\G(n)$ and $\G(k,n)$.
The \textbf{semilinear} group $\Gamma_n$ is the semidirect product (also called \textit{unrestricted wreath product}) of the group of field automorphisms $\textnormal{Aut}(\F_q)$ of $\F_q$ and $\M_n$:
\begin{equation*}
    \Gamma_n=\textnormal{Aut}(\F_q)\ltimes\M_n=\{(\sigma,M) \mid \sigma\in\textnormal{Aut}(\F_q), \, M\in\M_n\}.
\end{equation*}
The action of an element $(M,\sigma)$ on a vector $x\in\F_q^n$ is the component-wise application of $\sigma$ to~$x$, followed by the usual action of $M$.
This action too extends naturally to any Grassmannian and to the projective space.
We will not spell out these group actions, for which we refer the reader to \cite[Sections 1.6 and 1.7]{huffman2010fundamentals} and \cite[Sections 1.4 and 1.5]{betten2006error}.
It is readily checked that, through suitable embeddings, $\mathfrak{S}_n\subseteq\M_n\subseteq\Gamma_n$ and that if $q=p^h$ with $p$ a prime, we have
\begin{equation}\label{eq:groups_cardinality}
|\mathfrak{S}_n|=n!, \quad |\M_n|=n!(q-1)^n, \quad \text{and} \quad |\Gamma_n|=hn!(q-1)^n.
\end{equation}
When $h=1$, i.e. when $\F_q$ is a prime field, $\Gamma_n=\M_n$; when $q=2$, $\M_n=\mathfrak{S}_n$, and the corresponding types of equivalence coincide.
In particular, there is only one type of equivalence for binary linear codes.
The notations $\mathfrak{S}=\mathfrak{S}_n$, $\M=\M_n$ and $\Gamma=\Gamma_n$ will be preferred when~$n$ is clear from context.

It can be checked that a semilinear transformation does not change the Hamming weight of a vector, and hence all the groups mentioned above are groups of \textit{isometries} with respect to the Hamming metric.
By the MacWilliams Extension Theorem \cite[Section 7.9]{huffman2010fundamentals}, every linear  isometry of a code can be extended to a monomial transformation of the ambient space; this remains true for semilinear isometries and semilinear transformations \cite{hou2009asymptotic}.
It follows that the permutation/monomial/semilinear equivalence classes of linear codes are actually the \textit{orbits} of the actions of the corresponding groups on $\G(n)$.
Moreover, since the transformations are dimension-preserving, this also holds when considering the actions of the groups on $\G(k,n)$.
For a group $G$ acting on a set $X$, we denote by $X/G$ the set of orbits of $G$ in $X$.
With this in mind, we define the following quantities that are the main subject of this paper's work.
\begin{notation}
    Consider the action of $\M_n$ described above. 
    We introduce the quantities
    \begin{equation*}
    \Num_n^\M=|\G(n)/\M|\quad\textnormal{and}\quad\Num_{k(n),n}^\M=|\G(k,n)/\M|,
    \end{equation*}
    denoting respectively the number of monomial equivalence classes of codes, and the number of those classes having dimension $k=k(n)$.
    For $G=\mathfrak{S}_n$ (resp. $\Gamma_n$) we define the numbers \smash{$\Num_n^{\mathfrak{S}}$ and $\Num_{k(n),n}^{\mathfrak{S}}$ (resp. $\Num_n^{\Gamma}$ and $\Num_{k(n),n}^{\Gamma}$)} analogously.
\end{notation}

To avoid confusion, note that the numbers of equivalence classes in $\G(n)$ always have one index, while the numbers of inequivalent codes in $\G(k,n)$ always have two.
This paper mainly focuses on the latter ones.

\paragraph{Problem statement.}
As mentioned in the introduction, the asymptotic behaviour of the numbers $\Num_n^\mathfrak{S}$, $\Num_n^\M$ and $\Num_n^\Gamma$ is known~\cite{hou2005asymptotic,hou2007asymptotic,hou2009asymptotic}.
This paper studies the asymptotic behaviour of \smash{$\Num_{k(n),n}^\mathfrak{S}$, $\Num_{k(n),n}^\M$ and $\Num_{k(n),n}^\Gamma$} for fixed $q$ and $n\rightarrow\infty$.
We also consider the fraction that each of these numbers represents of the respective total number of equivalence classes. For example, in the monomial case we consider the quantity $\Num_{k(n),n}^\M/\Num_n^\M$.
This problem is of its own interest for coding theorists, but it also links to combinatorics and probability theory.

\section{Equivalence classes of codes with given dimension}\label{sec:asymptotic_number_eq_classes}

In this section we show how the number of inequivalent linear codes of dimension $k=k(n)$ is asymptotically related to the $q$-binomial coefficient \smash{$\binom{n}{k(n)}_q$} as $n\rightarrow\infty$. We will understand this 
for the class of functions $k: \N \to \N$ that satisfy a particular property, which we denote by $(\star)$ and define in~\Cref{prop:asymptotic_monomial_classes}.

The general strategy that we follow to tackle the problem of counting inequivalent codes with a given dimension relies on a standard group-theoretic argument: we recall it quickly to establish notation, and for completeness.
Let $G$ be a group acting on a set~$X$, where $g.x$ denotes the action of $g\in G$ on $x\in X$.
For any $g\in G$, we denote by~${\Fix(g,X)=\{x\in X:g.x=x\}}$ the set of $g$-invariant elements of $X$.
The kernel of the action is the set~$\Delta(G,X)=\{g\in G\mid\forall x\in X\;g.x=x\}$. 
By the Burnside Lemma then we have
\begin{equation}\label{eq:Burnside_lemma}
    |X/G|=\frac{1}{|G|}\sum_{g\in G}|\Fix(g,X)|=\frac{|\Delta(G,X)||X|}{|G|}+\sum_{g\in G\setminus\Delta(G,X)}|\Fix(G,X)|,
\end{equation}
and our results will rely on estimating the second sum on the RHS appropriately.
We are especially interested in applying this result to the groups $\mathfrak{S}_n$, $\M_n$, and $\Gamma_n$ acting on~$\G(k(n),n)$, similarly to what done in \cite{hou2005asymptotic,hou2009asymptotic} with their action on $\G(n)$.

We start by looking at monomial equivalence classes, and then consider permutation and semilinear classes as well.

\subsection{Monomial equivalence classes}
The starting point of our study is an application of the Burnside Lemma, which will be used repeatedly throughout the paper; see~\Cref{eq:Burnside_lemma} for the statement.
\begin{proposition}\label{prop:asymptotic_monomial_classes}
    There exist positive constants $A$, $\varepsilon$ such that, if $k: \N \to \N$ satisfies
    \begin{equation}\label{eq:star}
        \lim_{n\rightarrow\infty}\frac{1}{4}n^2-\varepsilon n+A\sqrt{n}-k(n)(n-k(n))=-\infty,\tag{$\star$}
    \end{equation}
    then 
    \begin{equation*}
        \Num_{k(n),n}^\M\sim \frac{\binom{n}{k(n)}_q}{n!(q-1)^{n-1}}, \quad\cfrac{\Num_{k(n),n}^\M}{\Num_n^\M}\sim p(k(n),n):=\cfrac{\binom{n}{k(n)}_q}{S(n)} \quad \mbox{as $n \to \infty$}.
    \end{equation*}
\end{proposition}
\begin{proof}
    For ease of notation, we write $k=k(n)$ and $\Delta=\Delta(\M,\G(k,n))=\{a\Id_n:a\in\F_q^*\}$, where $\Id_n$ denotes the $n\times n$ identity matrix. 
    Then $|\Delta|=q-1$ and the Burnside Lemma implies
    \begin{align*}
        \Num_{k,n}^\M&=\cfrac{(q-1)|\G(k,n)|+\sum_{M\in\M\setminus\Delta}|\Fix(M,\G(k,n))|}{|\M|}\\
        &=\cfrac{(q-1)\binom{n}{k}_q+\sum_{M\in\M\setminus\Delta}|\Fix(M,\G(k,n))|}{n!(q-1)^n}.
    \end{align*}
    Since $\G(k,n)\subseteq\G(n)$ we have
    \begin{equation*}    \cfrac{\sum_{M\in\M\setminus\Delta}|\Fix(M,\G(k,n))|}{\binom{n}{k}_q}\leq\cfrac{\sum_{M\in\M\setminus\Delta}|\Fix(M,\G(n))|}{q^{k(n-k)}}.
    \end{equation*}
    From \cite[Corollary 2.4 and Equation 4.1]{hou2005asymptotic} we know that there exist positive constants $A$, $\varepsilon$ such that, for $n$ large enough,    \begin{equation}\label{eq:sum_Fix_upper_bound}
\sum_{M\in\M\setminus\Delta}|\Fix(M,\G(n))|\in\bigO \left(q^{\frac{1}{4}n^2-\varepsilon n+A\sqrt{n}}\right).
    \end{equation}
If $\frac{1}{4}n^2-\varepsilon n+A\sqrt{n} -k(n-k)\to -\infty$, by \Cref{eq:sum_Fix_upper_bound} we have
    \begin{equation*}
        \cfrac{\sum_{M\in\M\setminus\Delta}|\Fix(M,\G(n))|}{q^{k(n-k)}}\in\smallo(1) \quad \mbox{as $n\rightarrow\infty$.}
    \end{equation*}
    It follows that
    \begin{equation*}
        \Num_{k,n}^\M=\cfrac{\binom{n}{k}_q(q-1+\smallo(1))}{n!(q-1)^n}\sim\cfrac{\binom{n}{k}_q}{n!(q-1)^{n-1}} \quad \mbox{as $n \to \infty$.}
    \end{equation*}
       The second asymptotic estimate in the statement  follows from~\cite[Theorem 4.1]{hou2005asymptotic}, since    \begin{equation}\label{eq:Hou_asymptotic_N}
        \Num_n^\M\sim\cfrac{S(n)}{n!(q-1)^{n-1}} \quad \mbox{as $n\rightarrow\infty$.} 
    \end{equation}
    This concludes the proof.
\end{proof}

\begin{notation}
  Throughout the remainder of the paper, 
  whenever condition $(\star)$ is mentioned, we implicitly mean that~$A$ and~$\varepsilon$ are the constants specified in the proof of the previous theorem.
\end{notation}

Note that condition $(\star)$ excludes many classes of functions $k(n)$ that one could find interesting: for example, $k(n)=\alpha\in\N$ or $k(n)=\lambda n$, $0<\lambda<1/2$, do not satisfy $(\star)$.
The following example outlines a class of functions $k(n)$ that \textit{do} satisfy $(\star)$, and that will be of interest in the forthcoming analysis.
\begin{example}\label{ex:k(n)=n/2-r}
Let $r$ be a constant and $k(n)=\lfloor n/2\rfloor-r$.
Then $k(n)$ satisfies $(\star)$, since
\begin{equation*}
    k(n)(n-k(n))=(\lfloor n/2\rfloor-r)(\lceil n/2\rceil+r)\geq\frac{1}{4}n^2-\frac{1}{4}-r^2-r.
\end{equation*}
A similar reasoning shows that $k(n)=\lceil n/2\rceil+r$ satisfies $(\star)$ as well. These
two functions will play a symmetric role in our analysis of the proportion of inequivalent codes having a given dimension (\Cref{sec:asymptotics_2}).
\end{example}
A more general class of functions $k(n)$ satisfying $(\star)$ is given in the following example.
\begin{example}
Generalising the previous example, a large class of functions $k(n)$ satisfying~$(\star)$ can be found as follows.
Write $k(n)=\lfloor n/2\rfloor-\ell(n)$, with $\ell(n)$ a positive function.
If $\ell(n)\in\smallo((\varepsilon n-A\sqrt{n})^{1/2})$, then $k(n)$ satisfies $(\star)$, as we have
\begin{equation*}
    k(n)(n-k(n))\geq\frac{1}{4}n^2-\left(\ell(n)+ \frac{1}{2} \right)^2.
\end{equation*}
For example, for $0\leq\alpha<1/2$, $\beta\in\R$, and $\ell(n)=n^\alpha\log n^\beta$, $k(n)$ satisfies $(\star)$.
\end{example}

The quantity $p(k,n)$, introduced in~\Cref{prop:asymptotic_monomial_classes} plays a role
also in the estimates for permutation and semilinear equivalence classes, as we see in the next subsection.

\subsection{Permutation and semilinear classes}
The analogue of \Cref{prop:asymptotic_monomial_classes} for permutation and semilinear equivalence classes of codes is the following result.
\begin{proposition}\label{prop:asymptotic_permutation_semilinear_classes}
    Let $q=p^h$, $p$ a prime, and assume that $k: \N \to \N$ satisfies $(\star)$. Then for the permutation equivalence classes of codes we have
    \begin{equation*}
        \Num_{k(n),n}^\mathfrak{S}\sim \frac{\binom{n}{k(n)}_q}{n!}, \quad \cfrac{\Num_{k(n),n}^\mathfrak{S}}{\Num_n^\mathfrak{S}}\sim p(k(n),n) \quad \mbox{as $n \to \infty$}.
    \end{equation*}   
    For semilinear classes we have
    \begin{equation*}
        \Num_{k(n),n}^\Gamma\sim \frac{\binom{n}{k(n)}_q}{hn!(q-1)^{n-1}},\quad \cfrac{\Num_{k(n),n}^\Gamma}{\Num_n^\Gamma}\sim p(k(n),n) \quad \mbox{as $n \to \infty$}.     
    \end{equation*}
\end{proposition}
\begin{proof}
    As in the proof of \Cref{prop:asymptotic_monomial_classes}, we let $k=k(n)$ and $\Delta=\Delta(\mathfrak{S},\G(k,n))$ for ease of notation. Regarding permutation classes, by the Burnside Lemma we have (notice that in this case $\Delta=\{\Id_n\}$)
    \begin{align*}
\Num_{k,n}^\mathfrak{S}&=\cfrac{|\G(k,n)|+\sum_{P\in\mathfrak{S}\setminus\Delta}|\Fix(P,\G(k,n))|}{|\mathfrak{S}|}=\cfrac{\binom{n}{k}_q+\sum_{P\in\mathfrak{S}\setminus\Delta}|\Fix(P,\G(k,n))|}{n!},
    \end{align*}
    and since $\mathfrak{S}\subseteq\M$ and $k$ satisfies $(\star)$, we have
    \begin{equation*}
        \cfrac{\sum_{P\in\mathfrak{S}\setminus\Delta}|\Fix(P,\G(k,n))|}{\binom{n}{k}_q}\leq\cfrac{\sum_{M\in\M\setminus\Delta}|\Fix(M,\G(k,n))|}{\binom{n}{k}_q}\in\smallo(1).
    \end{equation*}
    The rest of the proof is as in \Cref{prop:asymptotic_monomial_classes}, replacing the asymptotic estimate for~$\Num_n^\M$ of \Cref{eq:Hou_asymptotic_N} with the analogue result for $\Num_n^\mathfrak{S}$~\cite[Theorem 5.1]{hou2005asymptotic}, which tells us that    \begin{equation}\label{eq:Hou_asymptotic_N^S}
\Num_n^\mathfrak{S}\sim\frac{S(n)}{n!} \quad \mbox{ as $n\rightarrow\infty$}.
    \end{equation}

    Concerning semilinear classes, we apply the Brunside Lemma in a slightly different fashion. By letting $\Delta=\Delta(\Gamma,\G(k,n))$, we have
    \begin{align*}
        \Num_{k,n}^\Gamma&=\frac{1}{|\Gamma|}\sum_{M\in\M}|\Fix(M,\G(k,n))|+\frac{1}{|\Gamma|}\sum_{\gamma\in\Gamma\setminus\M}|\Fix(\gamma,\G(k,n))|\\
        &=\frac{1}{h}\Num_{k,n}^\M+\frac{1}{|\Gamma|}\sum_{\gamma\in\Gamma\setminus\M}|\Fix(\gamma,\G(k,n))|.
    \end{align*}
    We estimate the second summand on the RHS using the results of \cite[Section 2]{hou2009asymptotic}, from which it follows that there exists a positive constant $\varepsilon'$ with the property that, for $n$ large enough,  \begin{equation}\label{eq:sum_Fix_semilin_upper_bound}
        \sum_{\gamma\in\Gamma\setminus\M}|\Fix(\gamma,\G(n))|\in\bigO\left(q^{\frac{1}{4}n^2-\varepsilon' n^2}\right).
    \end{equation}
    Therefore,
    \begin{equation*}
        \cfrac{\sum_{\gamma\in\Gamma\setminus\M}|\Fix(\gamma,\G(k,n))|}{\binom{n}{k}_q}\leq\cfrac{\sum_{\gamma\in\Gamma\setminus\M}|\Fix(\gamma,\G(n))|}{q^{k(n-k)}}\in\bigO\left(q^{\frac{1}{4}n^2-\varepsilon' n^2-k(n-k)}\right).
    \end{equation*}
    Notice that since $k$ satisfies $(\star)$, we also have
    \begin{equation*}
        \lim_{n\rightarrow\infty}\frac{1}{4}n^2-\varepsilon' n^2-k(n-k)=-\infty,
    \end{equation*}
    hence
    \begin{align*}
        \Num_k^\Gamma\frac{hn!(q-1)^{n-1}}{\binom{n}{k}_q}&=\frac{1}{h}\Num_k^\M\frac{hn!(q-1)^{n-1}}{\binom{n}{k}_q}+\frac{hn!(q-1)^{n-1}}{|\Gamma|}\cfrac{\sum_{\gamma\in\Gamma\setminus\M}|\Fix(\gamma,\G(k,n))|}{\binom{n}{k}_q}\\
        &\sim 1+o(1)\sim1.
    \end{align*}
    The statement about the fraction $\Num_k^\Gamma/\Num_n^\Gamma$ follows from \cite[Equation 1.4]{hou2009asymptotic}, which tells us that     \begin{equation}\label{eq:Hou_asymptotic_N^Gamma}
        \Num_n^{\Gamma}\sim\cfrac{1}{h}\Num_n^\M \sim \cfrac{S(n)}{hn!(q-1)^{n-1}} \quad \mbox{as $n\rightarrow\infty$},
    \end{equation}
concluding the proof.
\end{proof}
\begin{remark}
   By comparing \Cref{prop:asymptotic_monomial_classes} and \Cref{prop:asymptotic_permutation_semilinear_classes} one sees that the three quantities
   $$\Num_{k(n),n}^\mathfrak{S}, \quad \Num_{k(n),n}^\M, \quad \Num_{k(n),n}^\Gamma$$ have different asymptotic behaviours, as expected from the definitions.
    Surprisingly instead, the fractions 
    \begin{equation*}
\cfrac{\Num_{k(n),n}^\mathfrak{S}}{\Num_n^\mathfrak{S}},\quad\cfrac{\Num_{k(n),n}^\M}{\Num_n^\M},\quad\cfrac{\Num_{k(n),n}^\Gamma}{\Num_n^\Gamma}
    \end{equation*}
    are all asymptotically equivalent to $p(k(n),n)$ for $n$ large.
    As we will see in \Cref{sec:asymptotics_2}, this enables a global asymptotic description of the proportion of inequivalent codes with a given dimension, regardless of the chosen notion of equivalence.
\end{remark}

\section{Asymptotics of the $q$-binomial and of  $\Num_{k(n),n}$}\label{sec:asymptotics_1}
This section computes the asymptotic number of equivalence classes of codes of dimension $k=k(n)$, where $k(n)$ satisfies property $(\star)$ and $n \to \infty$; see~\Cref{thm:asymptotic_Nk} below. To achieve so,
we start by studying the $q$-binomial coefficient $n$-choose-$k(n)$ and its asymptotic behaviour as $n$ grows.
Upper and lower bounds for the $q$-binomial coefficients are known: for every $q$ and $0\leq k\leq n$ we have
\begin{equation}\label{eq:q_binomial_bounds}
    q^{k(n-k)}\leq\binom{n}{k}_q\leq\frac{1}{K_q}q^{k(n-k)},
\end{equation}
where $K_q=\prod_{j=1}^{\infty}(1-q^{-j})$ is a finite constant depending only on $q$.
The lower bound is easy to see, while for the upper bound we refer to~\cite{gadouleau2008decoder}.
Recall that $K_q$ represents the fraction of $n\times n$ matrices over~$\F_q$ that are invertible as $n\rightarrow\infty$, and that ${K_q=\phi(q^{-1})}$, with $\phi$ the  \textit{Euler phi} function.
The bounds in \Cref{eq:q_binomial_bounds} already tell us that \smash{$\binom{n}{k}_q\in\bigO \left(q^{k(n-k)}\right)$}.

In the sequel, we are interested in determining a function $f_q: \N \to \R$ with the property that
\begin{equation}\label{eq:q-bin_asymptotic_f}
    \binom{n}{k(n)}_q\sim f_q(k(n))q^{k(n-k)} \quad \mbox{as $n \to \infty$}.
\end{equation}
Note that by this we do \textit{not} mean that the functions in \Cref{eq:q-bin_asymptotic_f} converge.
We begin by evaluating the ratio between two $q$-binomial coefficients.
For every $n$ and $q$, we refer to the $q$-binomial coefficient \smash{$\binom{n}{\lfloor n/2\rfloor}_q=\binom{n}{\lceil n/2\rceil}_q$} as the \textbf{central} $q$-binomial.
The following two lemmata are cornerstones of this paper.
\begin{lemma}\label{lem:asymptotic_q_binom_q_binom}
    We have
    \begin{equation*}
        \cfrac{\binom{n}{k}_q}{\binom{n}{\lfloor n/2\rfloor}_q}\sim \cfrac{K_q}{K_q(k)}\, q^{-(\lfloor n/2\rfloor-k)(\lceil n/2\rceil-k)} \quad \mbox{as $n \to \infty$},
    \end{equation*}
    where $K_q(k)=\prod_{j=1}^{k}(1-q^{-j})$ is the truncation of the product defining $K_q$ to $k$ terms.
    In particular,
    \begin{equation*}
        \cfrac{K_q}{K_q(k)}\sim\begin{cases}
            1&\text{if }\lim_{n\rightarrow\infty}k(n)=+\infty,\\
            \beta&\text{if }\lim_{n\rightarrow\infty}k(n)=\alpha <+\infty,\\
        \end{cases}
    \end{equation*}
    where $\alpha$ and $\beta$ are constants and $\beta<1$.
\end{lemma}
\begin{proof}
Let $\underm=\lfloor n/2\rfloor$, $\overm=\lceil n/2\rceil$; by symmetry of the $q$-binomial, we can assume $k\leq\underm$ without loss of generality.
Indeed, we always have $\min(k,n-k)\leq\underm$ and $\binom{n}{k}_q=\binom{n}{n-k}_q$.
Therefore,
\begin{align*}
    \cfrac{\binom{n}{k}_q}{\binom{n}{\lfloor n/2\rfloor}_q}&=\cfrac{\prod_{i=0}^{k-1}q^{n-i}-1}{\prod_{j=0}^{k-1}q^{k-j}-1}\cfrac{\prod_{j=0}^{\underm-1}q^{\underm-j}-1}{\prod_{i=0}^{\underm-1}q^{n-i}-1}=\cfrac{\prod_{i=0}^{k-1}q^{n-i}-1}{\prod_{j=1}^{k}q^{j}-1}\cfrac{\prod_{j=1}^{\underm}q^{j}-1}{\prod_{i=0}^{\underm-1}q^{n-i}-1}\\
    &=\cfrac{\prod_{j=k+1}^{\underm}q^j-1}{\prod_{i=k}^{\underm-1}q^{n-i}-1}=\cfrac{\prod_{j=1}^{\underm-k}q^{k+j}-1}{\prod_{i=\overm+1}^{n-k}q^{i}-1}=\cfrac{\prod_{j=1}^{\underm-k}q^{k+j}-1}{\prod_{i=1}^{\underm-k}q^{\overm+i}-1}\\
    &=\prod_{j=1}^{\underm-k}\frac{q^k}{q^{\overm}}\frac{q^j-q^{-k}}{q^j-q^{-\overm}}=q^{-(\overm-k)(\underm-k)}\prod_{j=1}^{\underm-k}\frac{1-q^{-(k+j)}}{1-q^{-(\overm+j)}},
\end{align*}
and we are left with proving that
\begin{equation*}
    \prod_{j=1}^{\underm-k}\frac{1-q^{-(k+j)}}{1-q^{-(\overm+j)}}\sim\frac{K_q}{K_q(k)}.
\end{equation*}
We have that $\frac{K_q}{K_q(k)}=\prod_{j=1}^{\infty}(1-q^{-(k+j)})$ and
\begin{equation*}
    \cfrac{\prod_{j=1}^{\infty}1-q^{-(k+j)}}{\prod_{j=1}^{\underm-k}\frac{1-q^{-(k+j)}}{1-q^{-(\overm+j)}}}=\prod_{j=1}^{\underm-k}(1-q^{-(\overm+j)})\prod_{j=1}^{\infty}(1-q^{-(\underm+j)}).
\end{equation*}
Note that
\begin{align*}
    1\geq\prod_{j=1}^{\underm-k}(1-q^{-(\overm+j)})\prod_{j=1}^{\infty}(1-q^{-(\underm+j)})&\geq\prod_{j=1}^{\infty}(1-q^{-(\overm+j)})\prod_{j=1}^{\infty}(1-q^{-(\underm+j)})\\
    &\geq\left(\prod_{j=1}^{\infty}(1-q^{-(\underm+j)})\right)^2.
\end{align*}
Taking the logarithm of the RHS we get
\begin{equation*}
    2\log\left(\prod_{j=1}^{\infty}(1-q^{-(\underm+j)})\right)=2\sum_{j=1}^\infty\log(1-q^{-(\underm+j)}).
\end{equation*}
For every fixed value of $j$ we have $\lim_{n\rightarrow\infty}\log(1-q^{-(\underm+j)})=0$ and
\begin{equation*}
    |\log(1-q^{-(\underm+j)})|\leq\frac{q^{\underm+j}}{q^{\underm+j}-1}-1=\frac{1}{q^{\underm+j}-1}\leq q^{-j}.
\end{equation*}
Since $\sum_{j=1}^{\infty}q^{-j}=(q-1)^{-1}<\infty$, we can swap limit and sum to obtain
\begin{equation*}
    \lim_{n\rightarrow\infty}2\sum_{j=1}^\infty\log(1-q^{-(\overm+j)})=2\sum_{j=1}^\infty\lim_{n\rightarrow\infty}\log(1-q^{-(\overm+j)})=0,
\end{equation*}
which implies $\lim_{n\rightarrow\infty}\prod_{j=1}^{\infty}(1-q^{-(\overm+j)})=1$ and
\begin{equation*}
    \prod_{j=1}^{\underm-k}\frac{1-q^{-(k+j)}}{1-q^{1-(\overm+j)}}\sim\frac{K_q}{K_q(k)},
\end{equation*}
which concludes our proof.
\end{proof}
In the following result, we study the asymptotic growth of the central $q$-binomial.
We separate this result from the previous lemma because we will use it independently also in the next section of the paper.
\begin{lemma}\label{lem:asymptotic_q_binom_power_q}
    We have
    \begin{equation*}
        \cfrac{\binom{n}{\lfloor n/2\rfloor}_q}{q^{\lfloor n/2\rfloor\lceil n/2\rceil}}\sim\cfrac{1}{K_q} \quad \mbox{as $n \to \infty$}.
    \end{equation*}
\end{lemma}
\begin{proof}
Let $\underm=\lfloor n/2\rfloor$, $\overm=\lceil n/2\rceil$. Then
\begin{align*}
    \cfrac{\binom{n}{\lfloor n/2\rfloor}_q}{q^{\lfloor n/2\rfloor\lceil n/2\rceil}}&=\cfrac{\prod_{i=0}^{\underm-1}\frac{q^{n-i}-1}{q^{\underm-i}-1}}{q^{\underm\overm}}=\prod_{i=0}^{\underm-1}\cfrac{q^{n-i}-1}{q^{\overm} (q^{\underm-i}-1)}\\
    &=\prod_{i=0}^{\underm-1}\cfrac{q^{\underm-i}-q^{-\overm}}{q^{\underm-i}-1}=\prod_{j=1}^{\underm}\cfrac{q^j-q^{-\overm}}{q^j-1},
\end{align*}
where $j=\underm-i$. It follows that
\begin{equation*}
    \cfrac{\binom{n}{\underm}_q}{K_q(\underm)q^{\underm\overm}}=\prod_{j=1}^{\underm}(1-q^{-(\overm+j)})\sim\prod_{j=1}^{\infty}(1-q^{-(\overm+j)}).
\end{equation*}
From the proof of \Cref{lem:asymptotic_q_binom_q_binom} we know that
\begin{equation*}
    \prod_{j=1}^{\infty}(1-q^{-(\overm+j)})\sim 1,
\end{equation*}
which gives the desired result.
\end{proof}
Part of the previous result can be obtained using \cite[Equation 6.2]{gruica2022common}, which implies that
\begin{equation*}
    \binom{2m}{m}_q\sim\cfrac{q^{m^2}}{K_q}.
\end{equation*}
In other words, \cite[Equation 6.2]{gruica2022common} can be used to show that the asymptotic result holds for the subsequence corresponding to even values of $n$, but not for the odd ones.

Combining the two lemmata we just proved, we obtain the following estimate for the asymptotic growth of the $q$-binomial coefficient.

\begin{corollary}\label{corol:asymptotic_q_binomial}
    We have
    \begin{equation*}
        \binom{n}{k(n)}_q\sim\cfrac{1}{K_q(k(n))}\,q^{k(n)(n-k(n))} \quad \mbox{as $n \to \infty$}.
    \end{equation*}
\end{corollary}
\begin{proof}
Write $k=k(n)$ for ease of notation.    By \Cref{lem:asymptotic_q_binom_q_binom,lem:asymptotic_q_binom_power_q} we have
    \begin{align*}
        \binom{n}{k}_q&=\cfrac{\binom{n}{k}_q}{\binom{n}{\lfloor n/2\rfloor}_q}\cfrac{\binom{n}{\lfloor n/2\rfloor}}{q^{\lfloor n/2\rfloor\lceil n/2\rceil}}q^{\lfloor n/2\rfloor\lceil n/2\rceil}\sim\cfrac{K_q}{K_q(k)K_q}q^{-(\lfloor n/2\rfloor-k)(\lceil n/2\rceil-k)}q^{\lfloor n/2\rfloor\lceil n/2\rceil}\\
        &\sim\cfrac{1}{K_q(k)}q^{k(\lfloor n/2\rfloor+\lceil n/2\rceil)-k^2}=\cfrac{1}{K_q(k)}q^{k(n-k)}. \qedhere
    \end{align*}
\end{proof}
\begin{remark}
For $k=k(n)$ the above corollary can be made more specific if $\alpha=\lim_{n\rightarrow\infty}k(n)$ exists.
If $\alpha<+\infty$, we have $K_q(k(n))\rightarrow K_q(\alpha)$, while if $\alpha=+\infty$ we have $K_q(k(n))\rightarrow K_q$.
In other words, we have $f_q(k(n))=1/K_q(\alpha)$ or $f_q(k(n))=1/K_q$ in \Cref{eq:q-bin_asymptotic_f}.
\end{remark}

The following theorem describes the asymptotic number of inequivalent codes of dimension $k = k(n)$ as $n \to \infty$.
It is one of the main results of this work and represents a contribution of fundamental nature to coding theory.

\begin{theorem}\label{thm:asymptotic_Nk}
    Let $q=p^h$, $p$ a prime, and assume $k(n)$ satisfies $(\star)$. Then for $n \to \infty$ we have
    \begin{equation}\label{eq:asymptotic_N_k_n_q}
        \Num^\mathfrak{S}_{k(n),n}\sim \cfrac{q^{k(n)(n-k(n))}}{K_qn!},\quad
        \Num^\M_{k(n),n}\sim\cfrac{q^{k(n)(n-k(n))}}{K_qn!(q-1)^{n-1}},\quad\Num^\Gamma_{k(n),n}\sim\cfrac{q^{k(n)(n-k(n))}}{K_qhn!(q-1)^{n-1}}.
    \end{equation}
\end{theorem}
\begin{proof}
    Apply~\Cref{corol:asymptotic_q_binomial} to the asymptotic results of~\Cref{prop:asymptotic_monomial_classes,prop:asymptotic_permutation_semilinear_classes} about the respective numbers of equivalence classes.
    Notice that, since $k(n)$ satisfies $(\star)$, we have $\lim_{n\rightarrow\infty}k(n)=+\infty$ and so $K_q(k(n))\sim K_q$.
\end{proof}

\begin{example}
    By \Cref{ex:k(n)=n/2-r}, we know that $k(n)=\lfloor n/2\rfloor -r$ satisfies $(\star)$. Therefore the asymptotic number of monomially inequivalent codes of dimension $k(n)$ satisfies
    \begin{equation*}
        \Num^\M_{k(n),n}\sim\cfrac{q^{\lfloor n/2\rfloor\lceil n/2\rceil -r^2}}{K_qn!(q-1)^{n-1}} \quad \mbox{as $n \to \infty$}.
    \end{equation*}
    By duality, this number should be equal to  the one we obtain for $k(n)=\lceil n/2\rceil+r$. It can be indeed checked that plugging this function into \Cref{eq:asymptotic_N_k_n_q} gives the same formulas.
\end{example}

\section{Asymptotics of $S(n)$ and of $\Num_{k(n),n}/\Num_n$}\label{sec:asymptotics_2}
We now turn to comparing the number of equivalence classes of codes with given dimension with the total number of equivalence classes of codes of \textit{any} dimension, for sufficiently large length.
As already shown in~\Cref{prop:asymptotic_monomial_classes,prop:asymptotic_permutation_semilinear_classes}, this comparison boils down to
investigating the asymptotic behaviour of the quantity~\smash{$p(k,n)=\binom{n}{k}_q/S(n)$} for fixed $q$, $k=k(n)$ a function of the length and $n\rightarrow\infty$.
We find out that, in general, one needs to consider two cases, given by the parity of $n$.
In other words, it is not possible to describe the asymptotic behaviour of $p(k,n)$ with a single function. To overcome this technical issue,
we look at the quantities $p^\textnormal{e}(k,m)=p(k,2m)$ and $p^\textnormal{o}(k,m)=p(k,2m+1)$ as $m\rightarrow\infty$, and we describe their asymptotic behaviour separately.
When $k=k(n)$ satisfies $(\star)$, we can apply the analysis to compute the asymptotic proportion of equivalence classes of codes that have a given dimension, as desired.

For every value of $m$, the sum (over $k$) of the positive numbers $p(k,2m)$ is 1, and the same holds for $p(k,2m+1)$.
This means they can be viewed as a probability distribution over $\Z$ (recall that \smash{$\binom{n}{k}_q=0$} for $k\in\Z\setminus[0,n]$).
Remarkably, and key for the results of this paper, when $k(n)$ is one of the functions in~\Cref{ex:k(n)=n/2-r}, our results show that these distributions have limit the Gaussian $\theta_3$ and $\theta_2$ distributions, respectively.
Since the functions of~\Cref{ex:k(n)=n/2-r} satisfy $(\star)$, this translates into an asymptotic description of the proportion of inequivalent codes that is particularly elegant.
The following lemma completes~\Cref{lem:asymptotic_q_binom_q_binom} and~\Cref{lem:asymptotic_q_binom_power_q} in forming the technical core of the paper.
It describes the fundamental difference between $n$ even and $n$ odd when computing the asymptotic of $S(n)$, by looking at its ratio with the central binomial.
This has two main consequences:
first, it allows to describe the asymptotics of the two subsequences of $p(k,n)$ corresponding to even and odd values of $n$.
Secondly, using this result we are able to describe the exact asymptotic behaviour of $S(n)$, a question left open in~\cite{wild2000asymptotic}.
Upper and lower bounds for $S(n)$ are known. For instance, we have (see~\cite{gadouleau2010packing})
\begin{equation}\label{eq:S(n)_bounds}
    q^{\lfloor n/2\rfloor\lceil n/2\rceil}\leq S(n)<\frac{\theta_3(q^{-1})+1}{K_q}q^{\lfloor n/2\rfloor\lceil n/2\rceil},
\end{equation}
where $K_q=\prod_{j=1}^{\infty}(1-q^{-j})$ already appeared in~\Cref{eq:q_binomial_bounds} and $\theta_3(\cdot)$ is the \textit{Jacobi} $\theta_3$ constant, which we both now define.
The \textbf{Jacobi} $\theta_2$ and $\theta_3$ constants are defined for
$0\leq w<1$
as
\begin{equation}\label{eq:theta_const}
    \theta_2(w)=\sum_{k=-\infty}^\infty w^{(k+1/2)^2}, \quad  \theta_3(w)=\sum_{k=-\infty}^\infty w^{k^2}.  
\end{equation}
In this paper, we are mainly interested in the values taken for $w=q^{-1}$.
We can already see from~\Cref{eq:S(n)_bounds} that the $\theta_3$ constant plays a role in bounding $S(n)$ from above; we will actually show that the $\theta_2$ and $\theta_3$ play a role in determining the exact asymptotic behaviour of $S(n)$.
\begin{lemma}\label{lem:asymptotic_sum_q_binom}
    We have
    \begin{equation*}
        \cfrac{\binom{n}{\lfloor n/2\rfloor}_q}{S(n)}=\begin{cases}
            \cfrac{\binom{2m}{m}_q}{S_q(2m)}\sim\cfrac{1}{\theta_3(q^{-1})}&\mbox{if $n=2m$ and $m \to \infty$},\\
            \cfrac{\binom{2m+1}{m}_q}{S_q(2m+1)}\sim\cfrac{1}{q^{1/4}\theta_2(q^{-1})}&\mbox{if $n=2m+1$ and $m \to \infty$}.
        \end{cases}
    \end{equation*}
\end{lemma}
\begin{proof}
    We first look at the case $n=2m$.
    Define a sequence of functions
\begin{equation*}
    f_m(r)=\cfrac{\binom{2m}{m-r}_q}{\binom{2m}{m}_q}.
\end{equation*}
By \Cref{lem:asymptotic_q_binom_q_binom} we have $\lim_{m\rightarrow\infty}f_m(r)=f(r):=q^{-r^2}$.
Moreover for every $r$ we have $0\leq f_m(r)\leq q^{-|r|}$ and 
\begin{equation}\label{eq:sum_1/q}
    \sum_{r=-\infty}^{+\infty}q^{-|r|}=\frac{3q-1}{q-1}<+\infty,
\end{equation}
where we used $q\geq2$. Thus by the Dominated Convergence Theorem,
\begin{align*}
    \lim_{m\rightarrow\infty}\cfrac{S(2m,q)}{\binom{2m}{m}_q}&=\lim_{m\rightarrow\infty}\sum_{r=-m}^{m}\cfrac{\binom{2m}{m-r}_q}{\binom{2m}{m}_q}=\lim_{m\rightarrow\infty}\sum_{r=-\infty}^{\infty}\cfrac{\binom{2m}{m-r}_q}{\binom{2m}{m}_q}\\
    &=\sum_{r=-\infty}^{\infty}\lim_{m\rightarrow\infty}\cfrac{\binom{2m}{m-r}_q}{\binom{2m}{m}_q}=\sum_{r=-\infty}^{\infty}\cfrac{1}{q^{r^2}}=\theta_3(q^{-1}),
\end{align*}
which is equivalent to our statement.

For $n=2m+1$ the proof is similar: for every $r$ we define a sequence of functions
\begin{equation*}
    f_m(r)=\cfrac{\binom{2m+1}{m-r}_q}{\binom{2m+1}{m}_q}.
\end{equation*}
Then by \Cref{lem:asymptotic_q_binom_q_binom}, $f_m(r)$ converges pointwise to ${f(r)=q^{-r(r+1)}}$ as $m\rightarrow\infty$.
Moreover, for every $r$ we have $0\leq f_m(r)\leq q^{-|r|}$. Again by \Cref{eq:sum_1/q} we have
\begin{align*}
    \lim_{m\rightarrow\infty}\cfrac{S(2m+1,q)}{\binom{2m+1}{m}_q}&=\lim_{m\rightarrow\infty}\sum_{r=-m}^{m+1}\cfrac{\binom{2m+1}{m-r}_q}{\binom{2m+1}{m}_q}=\lim_{m\rightarrow\infty}\sum_{r=-\infty}^{+\infty}\cfrac{\binom{2m+1}{m-r}_q}{\binom{2m+1}{m}_q}\\
    &=\sum_{r=-\infty}^{+\infty}\lim_{m\rightarrow\infty}\cfrac{\binom{2m}{m-r}_q}{\binom{2m}{m}_q}=\sum_{r=-\infty}^{\infty}\cfrac{1}{q^{r(r+1)}}=q^{1/4}\theta_2(q^{-1}).
\end{align*}
Taking reciprocals concludes our proof of the second part of the statement.
\end{proof}

\Cref{lem:asymptotic_sum_q_binom} allows us to isolate, and account for, the difference between $n$ even and $n$ odd when looking at the asymptotic behaviour of $p(k,n)$.
This study is naturally completed using~\Cref{lem:asymptotic_q_binom_q_binom}, as we now illustrate.
\begin{theorem}\label{thm:asymptotic_pe_po}
    We have
    \begin{equation*}
        p^\textnormal{e}(k,m)\sim\cfrac{K_q}{K_q(k)\theta_3(q^{-1})}q^{-(m-k)^2}, \quad p^\textnormal{o}(k,m)\sim\cfrac{K_q}{K_q(k)\theta_2(q^{-1})}q^{-(m-k+1/2)^2} \quad \mbox{as $m\rightarrow\infty$}.
    \end{equation*}
\end{theorem}
\begin{proof}
    We have
    \begin{equation*}
        p^\textnormal{e}(k,m)=p(k,2m)=\cfrac{\binom{2m}{m}_q}{S_q(2m)}\cfrac{\binom{2m}{k}_q}{\binom{2m}{m}_q}\sim\cfrac{1}{\theta_3(q^{-1})}\cfrac{K_q}{K_q(k)}q^{-(m-k)^2}
    \end{equation*}
    by \Cref{lem:asymptotic_q_binom_q_binom} and \Cref{lem:asymptotic_sum_q_binom}. The proof for $p^\textnormal{o}(k,m)$ follows the same steps.
\end{proof}

When $k=k(n)$ satisfies $(\star)$, our results allow for the description of the asymptotic proportion of inequivalent codes of dimension $k$, which is one of the centerpieces of this paper.

\begin{corollary}\label{corol:asymptotic_fraction_eq_classes}
    Assume that $k=k(n)$ satisfies $(\star)$.
    For $n=2m$, $m\rightarrow\infty$, we have
    \begin{equation*}
        \frac{\Num^\mathfrak{S}_{k(2m),2m}}{\Num^\mathfrak{S}_{2m}}\sim\frac{\Num^\M_{k(2m),2m}}{\Num^\M_{2m}}\sim\frac{\Num^\Gamma_{k(2m),2m}}{\Num^\Gamma_{2m}}\sim\cfrac{1}{\theta_3(q^{-1})}q^{-(m-k(2m))^2}. 
    \end{equation*}
    For $n=2m+1$, $m\rightarrow\infty$, we have 
    \begin{equation*}
        \frac{\Num^\mathfrak{S}_{k(2m+1),2m+1}}{\Num^\mathfrak{S}_{2m+1}}\sim\frac{\Num^\M_{k(2m+1),2m+1}}{\Num^\M_{2m+1}}\sim\frac{\Num^\Gamma_{k(2m+1),2m+1}}{\Num^\Gamma_{2m+1}}\sim\cfrac{1}{\theta_2(q^{-1})}q^{-(m-k(2m+1)+1/2)^2}. 
    \end{equation*}
\end{corollary}
\begin{proof}
    Apply~\Cref{thm:asymptotic_pe_po}, noticing that if $k(n)$ satisfies $(\star)$,  then $\lim_{n\rightarrow\infty}k(n)=+\infty$ and therefore $K_q/K_q(k(n))\sim1$.
\end{proof}

Of particular interest in this paper
is the case $k(n)=\lfloor n/2\rfloor-r$ for some fixed $r\in\N$.
For every $n$ and $q$, we extend the definition of $p(k,n)$ to every $k\in\Z$ by setting $p(k,n)=0$ whenever $k\notin[0,n]$.
Since $0\leq p(k,n)\leq1$ for every $k\in\Z$ and $\sum_{k\in\Z}p(k,n)=1$, 
the $p(k,n)$'s define a discrete probability distribution over $\Z$ via $\Prob(k)=p(k,n)$.
We then consider the following shifts of the distributions:
\begin{enumerate}
    \item if $n=2m$, define a distribution on $\Z$ by $\Prob^{\textnormal{e}}_m(r)=p_q^{\textnormal{e}}(m-r,m)$;
    \item if $n=2m+1$ is odd, define a distribution on $1/2+\Z$ by $\Prob^{\textnormal{o}}_m(r)=p_q^{\textnormal{o}}(m-r+1/2,m)$.
\end{enumerate}

The two shifted distributions are symmetric with respect to 0, meaning $\Prob^{\textnormal{e}}_m(r)=\Prob^{\textnormal{e}}_m(-r)$ and $\Prob^{\textnormal{o}}_m(r)=\Prob^{\textnormal{o}}_m(-r)$ for every $r$ and $m$.
Informally, one can see $r$ as a measure of the distance from the centre of the distribution.
One of the main findings of this paper is that, as $m\rightarrow\infty$, the two distributions converge pointwise to the \textit{discrete Gaussian} $\theta_3$ and $\theta_2$ distributions with \textit{nome} $1/q$, studied in \cite{salminen2024probabilistic} in connection to the Brownian motion.
\begin{definition}\label{def_discr_gaussian}
    Let $w\in\R$, $0<w<1$, be a constant.
    The \textbf{discrete Gaussian $\theta_2$-distribution} is defined by the density 
    \begin{equation*}
        \Prob_{\theta_2}(k)=\cfrac{w^{k^2}}{\theta_2(w)},\quad k\in\frac{1}{2}+\Z,
    \end{equation*}
    whereas the \textbf{discrete Gaussian $\theta_3$-distribution} is defined by the density
    \begin{equation*}
        \Prob_{\theta_3}(k)=\cfrac{w^{k^2}}{\theta_3(w)}, \quad k\in\Z.
    \end{equation*}
    The quantity $w$ is called the \textbf{nome} of the distributions.
\end{definition}

It was shown in~\cite{kemp1997characterizations} that the Gaussian $\theta_3$ distribution is the \textit{maximum entropy distribution} on $\Z$ having a  specified mean and variance.
This property qualifies the distribution as a discrete counterpart of the Gaussian distribution, which has the same characterisation over~$\R$.
\begin{remark}
    In \cite{salminen2024probabilistic}, the $\theta_2$ Gaussian is defined to take values in $\Z$ instead of $1/2+\Z$. We shift the domain to have a distribution that is symmetric around 0.
    This also has the advantage of yielding a more concise formulation for the exponents of the nome.
\end{remark}
We are interested in the case where the nome is $w=1/q$.
The following corollary spells out the covergence of the distributions $\Prob^{\textnormal{e}}$ and $\Prob^{\textnormal{o}}$.
It is a straightforward consequence of our previous results, but nonetheless one of the most relevant findings of this work.
\begin{corollary}\label{corol:limit_distr}
    As $m\rightarrow\infty$ we have the following convergences in distribution:
    \begin{equation*}
        \Prob^{\textnormal{e}}_m\rightarrow\Prob_{\theta_3}\quad\textnormal{and}\quad\Prob^{\textnormal{o}}_m\rightarrow\Prob_{\theta_2},
    \end{equation*}
    where the nome of the limit distirbutions is $1/q$.
\end{corollary}
\begin{proof}
    For every fixed $r\in\Z$, by~\Cref{thm:asymptotic_pe_po} and~\Cref{corol:asymptotic_fraction_eq_classes}, we have
    \begin{equation*}
        \Prob^{\textnormal{e}}_m(r)=p_m^{\textnormal{e}}(m-r,m)\sim \cfrac{q^{-r^2}}{\theta_3(1/q)}=\Prob_{\theta_3}(r),
    \end{equation*}
    and the result follows from the characterization of convergence in distribution in terms of pointwise convergence; see for instance~\cite{jacod2004probability}.
    The proof for $\Prob^{\textnormal{o}}_m$ is analogous.
\end{proof}
Stochastic characterisations for Gaussian $\theta_2$ and $\theta_3$ distributed random variables were proposed in \cite{salminen2024probabilistic}.
These descriptions are based on infinite product representations of the~$\theta_2$ and~$\theta_3$ Jacobi theta functions, and involve the sum of infinitely many Bernoulli random variables with different distributions.
A different characterisation for the Gaussian $\theta_3$ distribution as the difference of Heine distributions was proposed in~\cite{kemp1997characterizations}.
\Cref{corol:limit_distr} offers an alternative result in this sense: as $m$ grows, the distributions $\Prob^{\textnormal{e}}_m(r)$ (resp. $\Prob^{\textnormal{o}}_m(r)$) become increasingly good approximations of the Gaussian $\theta_3$ (resp. $\theta_2$), providing also an effective way to compute the values of the probabilities.

\paragraph{Answering an open question from~\cite{wild2000asymptotic}.}
Our results also allow for the description of the asymptotic behaviour of $S(n)$.
It was shown in \cite[Lemma 1]{wild2000asymptotic} that for every $q$ there exist constants $d_1=d_1(q)$ and $d_2=d_2(q)$ such that
\begin{equation*}
    S(2m+1)\sim d_1q^{(2m+1)^2/4},\quad
     S(2m)\sim d_2q^{(2m)^2/4} \quad \mbox{as $m \to \infty$}.
\end{equation*}
In the same work, it is shown that $d_1<1\leq d_2$ for $q\geq 49$, and that $d_1<d_2$ for all $q<49$, implying that the two values never coincide. Yet, the two numbers are not computed explicitly.
 From~\Cref{eq:S(n)_bounds} it is evident that $d_2\leq\frac{\theta_3(1/q)+1}{K_q}$.
Furthermore, our results imply closed formulas for the constants $d_1$ and $d_2$ as in the following result.
\begin{corollary}
    We have
    \begin{equation*}
        S(2m)\sim\frac{\theta_3(1/q)}{K_q}q^{m^2}\quad\textnormal{and}\quad S(2m+1)\sim\frac{\theta_2(1/q)}{K_q}q^{(m+1/2)^2}
    \end{equation*}
\end{corollary}
\begin{proof}
    Simply combine~\Cref{lem:asymptotic_sum_q_binom} and~\Cref{lem:asymptotic_q_binom_power_q}.
\end{proof}
Therefore, in the notation of~\cite{wild2000asymptotic}, we have \smash{$d_1=\frac{\theta_2(1/q)}{K_q}$} and \smash{$d_2=\frac{\theta_3(1/q)}{K_q}$}.

\end{sloppypar}

\bibliographystyle{abbrv}
\bibliography{refs}

\end{document}